\documentclass[submission,copyright,creativecommons]{eptcs}
\providecommand{\event}{WoF 2015} 

\usepackage[utf8]{inputenc}
\usepackage{amssymb}
\usepackage{amsmath}
\usepackage{amsthm}
\usepackage{stmaryrd}
\usepackage{xcolor}
\usepackage{proof}
\usepackage{comment}
\usepackage{xspace}
\usepackage{url}
\usepackage{comment}

\newcommand\ie{\emph{i.e.}\xspace}

\newtheorem{definition}{Definition}
\newtheorem{lemma}{Lemma}
\newtheorem{theorem}{Theorem}
\newtheorem{conjecture}{Conjecture}

\title{Towards the Automated Generation of Focused Proof Systems}

\author{Vivek Nigam
\institute{Federal University of Para\'{i}ba, Brazil}
\email{vivek.nigam@gmail.com}
\and
Giselle Reis
\institute{Inria \& LIX, France}
\email{giselle.reis@inria.fr}
\and
Leonardo Lima
\institute{Federal University of Para\'{i}ba, Brazil}
\email{leonardo.alfs@gmail.com}
}

\newcommand\authorrunning{Nigam, Reis and Lima}
\newcommand\titlerunning{Towards the Automated Generation of Focused Proof Systems}

\begin{document}
\maketitle

\begin{abstract}
This paper tackles the problem of formulating and proving the completeness of
focused-like proof systems in an automated fashion. Focusing is a discipline on
proofs which structures them into phases in order to reduce proof search
non-determinism. We demonstrate that it is possible to construct a
complete focused proof system from a given un-focused proof system if it
satisfies some conditions. Our key idea is to generalize the completeness proof
based on permutation lemmas given by Miller and Saurin for the focused linear
logic proof system. This is done by building a graph from the rule permutation
relation of a proof system, called permutation graph. We then show that from the
permutation graph of a given proof system, it is possible to construct a
complete focused proof system, and additionally infer for which formulas
contraction is admissible. An implementation for building the permutation
graph of a system is provided. We apply our technique to generate
the focused proof systems MALLF, LJF and LKF for linear, intuitionistic and
classical logics, respectively.
\end{abstract}


\section{Introduction}


In spite of its widespread use, the proposition and completeness
proofs of focused proof systems are still an \emph{ad-hoc} and hard task, done
for each individual system separately. For example, the original completeness
proof for the focused linear logic proof system (LLF)~\cite{andreoli92jlc} is
very specific to linear logic. The completeness proof for many focused proof
systems for intuitionistic logic, such as LJF~\cite{liang07csl}, LKQ and
LKT~\cite{danos93wll}, are obtained by using non-trivial encodings of
intuitionistic logic in linear logic.

One exception, however, is the work of Miller and Saurin~\cite{miller07cslb},
where they propose a modular way to prove the completeness of focused proof
systems based on permutation lemmas and proof transformations. They show that a
given focused proof system is complete with respect to its unfocused version by
demonstrating that any proof in the unfocused system can be transformed into a
proof in the focused system. Their proof technique has been successfully adapted
to prove the completeness of a number of focused proof systems based on
linear logic, such as ELL~\cite{saurin08phd}, $\mu$MALL~\cite{baelde08phd} and
SELLF~\cite{nigam09phd}.

This paper proposes a method for the automated generation of a sound and complete focused
proof system from a given unfocused sequent calculus proof system. Our approach
uses as theoretical foundations the modular proof given by Miller and
Saurin~\cite{miller07cslb}. There are, however, a number of challenges in
automating such a proof for any given unfocused proof system:
%
  (1) Not all proof systems seem to admit a focused version.
  We define sufficient conditions based on the definitions in \cite{miller07cslb};
%
  (2) Even if a proof system satisfies such conditions, there are many design 
  choices when formulating a focused version for a system;
%
  (3) Miller and Saurin's proof cannot be directly applied to proof
  systems that have contraction and weakening rules, such as LJ;
  Focused proof systems, such as LJF and LKF, allow only the contraction of some
  formulas. This result was obtained by non-trivial encodings in linear
  logic~\cite{liang09tcs}. Here, we demonstrate that this can be
  obtained in the system itself, \ie, without a detour through linear logic;
  (4) Miller and Saurin did not formalize why their procedure or transforming
  an unfocused proof into a focused one terminates. 
  It already seems challenging to do so for MALL as permutations are not necessarily
  size preserving (with respect to the number of inferences). 
  We are still investigating general conditions and this is left to future work.

In order to overcome these challenges, we introduce in Section~\ref{sec:focusing} 
the notion of permutation graphs. 
Our previous work~\cite{nigam13iclp,nigam14ijcar} showed how to check
whether a rule permutes over another in an automated fashion. We use these
results to construct the permutation graph of a proof system. This paper then
shows that, by analysing the permutation graph of an unfocused proof system, we
can construct possibly different focused versions of this system, all sound and 
complete (provided a proof of termination is given). 
We sketch in Section~\ref{sec:contraction} how to check the admissibility of contraction rules.

\section{Permutation Graphs}
\label{sec:perm_graphs}


In the following we assume that we are given a sequent calculus proof system
$\mathbb{S}$ which is commutative, \ie, sequents are formed by multi-sets
of formulas, and whose non-atomic initial and cut rules are admissible. 
%
%
%
We will also assume
that whenever contraction is allowed then weakening is also allowed, that is,
our systems can be affine, but not relevant. Finally, we assume the reader is 
familiar with basic proof theory terminology, such as main and auxiliary formulas, 
formula ancestors. 





\begin{definition}[Permutability]
Let $\alpha$ and $\beta$ be two inference rules in a sequent calculus system
$\mathbb{S}$. 
We will say that $\alpha$ \emph{permutes up} $\beta$, denoted by
$\alpha \uparrow \beta$, if for every $\mathbb{S}$ derivation of a sequent
$\mathcal{S}$ in which $\alpha$ operates on $\mathcal{S}$ and $\beta$ operates
on one or more of $\alpha$'s premises (but not on auxiliary formulas of $\alpha$), 
%
there exists another $\mathbb{S}$ derivation of $\mathcal{S}$ in which $\beta$
operates on $\mathcal{S}$ and $\alpha$ operates on zero or more of $\beta$'s
premises (but not on $\beta$'s auxiliary formulas).
Consequently, $\beta$ \emph{permutes down} $\alpha$ ($\beta \downarrow \alpha$). 
\end{definition}

Note that if there is no derivation in which $\beta$ operates on $\alpha$'s
premises without acting on its auxiliary formulas (e.g., $\vee_r$ and $\wedge_r$
in LJ), the permutation holds vacuously.
%


\begin{definition}[Permutation graph]
Let $\mathcal{R}$ be the set of inference rules of a sequent calculus system
$\mathbb{S}$. We construct the (directed) \emph{permutation graph} $P_{\mathbb{S}}=(V, E)$
for $\mathbb{S}$ by taking $V = \mathcal{R}$ and $E = \{ (\alpha, \beta) \;|\;
\alpha \uparrow \beta \}$.
\end{definition}
%

\begin{definition}[Permutation cliques]
Let $\mathbb{S}$ be a sequent calculus system and $P_\mathbb{S}$ its permutation
graph. Consider $P^*_\mathbb{S}=(V^*, E^*)$ the undirected graph obtained from
$P_\mathbb{S}=(V,E)$ by taking $V^* = V$ and $E^* = \{ (\alpha, \beta) \;|\;
(\alpha, \beta) \in E \text{ and } (\beta, \alpha) \in E \}$. Then the
\emph{permutation cliques} of $\mathbb{S}$ are the maximal cliques\footnote{A
clique in a graph $G$ is a set of vertices such that all vertices are pairwise
connected by one edge.} of $P^*_\mathbb{S}$.
\end{definition}

For LJ, we obtain the following cliques
$
CL_1 = \{ \wedge_l, \vee_l, \rightarrow_r, \vee_r \} 
$ and
$
CL_2 = \{ \wedge_r, \vee_r, \rightarrow_l \}
$.

Permutation cliques can be thought of as equivalence classes for inference rules. 
For example, the rule $\wedge_l$ permutes over all rules in $CL_1$. Permutation 
cliques are not always disjoint. For example, the rule $\vee_r$ appears in both 
cliques.

\begin{definition}[Permutation partition]
Let $\mathbb{S}$ be a proof system and $P_\mathbb{S}$ its permutation graph.
Then a \emph{permutation partition} $\mathcal{P}$ is a partition of
$P_\mathbb{S}$ such that each component is a complete graph. We will call each
component of such partitions a \emph{permutation component}, motivated by the
fact that inferences in the same component permute over each other.
\end{definition}

It is always possible to find such a partition by taking each component to be
one single vertex, but we are mostly interested in bi-partitions.

Although in general cliques are computed in exponential time, it is still
feasible to compute them since the permutation graph is usually small. The
partitions can be obtained simply by choosing at most one partition to those
rules present in more than one clique. Therefore, there might 
be many possible ways to partition the rules of a system. In what
follows (Definition~\ref{def:focusable_partition}) we will define which are the 
partitions that will yield a focused proof system. As we will see, the following 
partition will lead to LJF, restricted to multiplicative conjunctions:  
$
C_1 = \{ \wedge_l, \vee_l, \rightarrow_r \}
$ and
$
C_2 = \{ \wedge_r, \vee_r, \rightarrow_l \}
$.

\begin{definition}[Permutation partition hierarchy]
Let $\mathbb{S}$ be a proof system, $P_\mathbb{S}$ its permutation graph and
$\mathcal{P} = C_1, ..., C_n$ a permutation partition. We say that $C_i \downarrow C_j$
iff for every inference $\alpha_i \in C_i$ and $\alpha_j \in C_j$ we have that
$\alpha_i \downarrow \alpha_j$, \ie, $\alpha_j \uparrow \alpha_i$ or
equivalently $(\alpha_j, \alpha_i) \in P_\mathbb{S}$.
\end{definition}

Notice that the partition hierarchy can be easily computed from the permutation graph. For the partition used above, we have $C_1 \downarrow C_2$.

\section{Focused Proof Systems Generation}
\label{sec:focusing}

We derive a focused proof system $\mathbb{S}^f$ from the permutation partitions of a given proof system $\mathbb{S}$ if some conditions
are fulfilled. In this section we explain these conditions and prove that the
induced focused system is sound and complete with respect to $\mathbb{S}$.

\begin{definition}[Focusable permutation partition]
\label{def:focusable_partition}
Let $\mathbb{S}$ be a sequent calculus proof system and $C_1, ..., C_n$ a
permutation partition of the rules in $\mathbb{S}$. We say that it is a
\emph{focusable permutation partition} if: 
\begin{itemize}
  \item $n = 2$ and $C_1 \downarrow C_2$;
  %
  \item Every rule in component $C_2$ has at most one auxiliary formula in each
  premise;
  \item Every non-unary rule in component $C_2$ splits the context among the
  premises (\ie, there is no implicit copying of context formulas on
  branching rules).
\end{itemize}
We call $C_1$ the negative component and $C_2$ the positive component following
usual terminology from the focusing literature (e.g. \cite{liang07csl}) and
classify formula occurrences in a proof as negative and positive according to
their \underline{introduction rules}.
\end{definition}

Observe that, in contrast to the usual approach, we do not assign polarities to
connectives on their own. Therefore the polarity of a formula can change
depending on whether it occurs on the right or on the left side of the sequent.
%
As for now, we will only define permutation partitions of logical inference
rules. The structural inference rules will be treated separately. In particular, 
the role of contraction and its relation to the partitions is discussed in 
Section \ref{sec:contraction}. 

The partition $\{C_1, C_2\}$ for LJ is a focusable permutation partition. 
Interestingly, LJ allows for other focusable partitions, for example:
$
C_1 = \{ \wedge_l, \vee_l, \rightarrow_r, \vee_r \}
$ and
$
C_2 = \{ \wedge_r, \rightarrow_l \}
$.

We conjecture that all proof systems derived from a focusable permutation
partition are sound and complete.
It is not our 
goal here to justify which partition leads to a more suitable focused proof 
system, as this would depend on the context where the proof system would be used. 


Based on the focusable permutation partition, we can define a focused proof
system for $\mathbb{S}$. This definition is syntactically different from those
usually present in the literature. It will, in particular, force the store and
subsequent selection of a negative formula. This extra step is only for the sake
of uniformity and clear separation between phases (there will always be a ``no
phase'' state between two phases). 

\begin{definition}[Focused proof system]
\label{def:focused-pf}
Let $\mathbb{S}$ be a sequent calculus proof system and $C_1 \downarrow C_2$
a focusable permutation partition of the rules in $\mathbb{S}$. Then we
can define the \emph{focused} system $\mathbb{S}^f$ in the following way:

\textbf{Sequents.} 
$\mathbb{S}^f$ sequents are of the shape $\Gamma ; \Gamma' \vdash^p \Delta ; \Delta'$,
where $p \in \{ +, -, 0 \}$ indicates a positive, negative and neutral polarity sequents
respectively. We will call $\Gamma'$ and $\Delta'$ the \emph{active} contexts.

\textbf{Inference Rules.}
For each rule $\alpha$ in $\mathbb{S}$ belonging to the negative (positive)
component, $\mathbb{S}^f$ will have a rule $\alpha$ with conclusion and premises
being negative (positive) sequents and main and auxiliary formulas occurring in
the active contexts.

\textbf{Structural rules.}
The connection between the phases is done via the following structural rules.

\emph{Selection rules} move a formula $F$ to the active context. If $F$ is
negative, then $p = -$. If $F$ is positive, then there is no negative $F' \in
\Gamma \cup \Delta$ and $p = +$. 
\emph{Store rules} remove a formula $F$ from the active context if $F$ is negative and $p = +$ or if $F$ is positive and $p = -$. 
The \emph{end rule} removes the label $p = \{ +, - \}$ of a sequent by setting it to 0
if the active contexts are empty.
{\footnotesize
\[
\infer[sel_l]{\Gamma, F ; \cdot \vdash^0 \Delta ; \cdot}{\Gamma ; F \vdash^p \Delta ; \cdot}
\qquad
\infer[sel_r]{\Gamma ; \cdot \vdash^0 \Delta, F ; \cdot}{\Gamma ; \cdot \vdash^p \Delta ; F}
\qquad
\infer[st_l]{\Gamma ; \Lambda, F \vdash^p \Delta ; \Pi}{\Gamma, F; \Lambda \vdash^p \Delta ; \Pi}
\qquad
\infer[st_r]{\Gamma ; \Lambda \vdash^p \Delta ; \Pi, F}{\Gamma ; \Lambda \vdash^p \Delta, F ; \Pi}
\qquad
\infer[end]{\Gamma ; \cdot \vdash^p \Delta ; \cdot}{\Gamma ; \cdot \vdash^0 \Delta ; \cdot}
\]
}%

\noindent
An $\mathbb{S}^f$ proof is characterized by sequences of inferences
labeled with $+$ or $-$ which we will call \emph{phases}. Thus, we
can say that selection rules are responsible for starting a phase and the end
rule finishes a phase. Between any two phases there is always a
``neutral'' state, denoted by a sequent labeled with 0.
\end{definition}

We can prove using the machinery given in \cite{miller07cslb} that the focused proof 
system obtained is complete. There is one catch, however: 
one also needs to prove that the procedure to convert an unfocused proof into a focused proof
using permutation lemmas terminates. This was not formalized in \cite{miller07cslb}, although
one can prove it. Finding general conditions is more challenging and is subject of current
investigation.

\begin{conjecture}[Completeness of focused proof systems]
\label{thm:completeness}
A sequent $\Gamma \vdash \Delta$ is provable in $\mathbb{S}$ iff the sequent
$\Gamma; \cdot \vdash^0 \Delta; \cdot$ is provable in $\mathbb{S}^f$.
\end{conjecture}

\section{Admissibility of contraction}
\label{sec:contraction}

During proof search, it is desirable to avoid unnecessary copying of formulas
at each rule application. Either by not copying the same context in all premises
or by not auto-contracting the main formula of a rule application. The analysis
of where the contraction rule lies in the permutation cliques can give us
insights on when it can be avoided. 

\begin{definition}[Admissibility of contraction]
Let $\mathbb{S}$ be a sequent calculus system with a set of rules $\mathcal{R}$.
We say that contraction is \emph{admissible} for $\mathcal{R}' \subseteq
\mathcal{R}$ if for every $\mathbb{S}$ derivation $\varphi$ there exists an
$\mathbb{S}$ derivation $\varphi'$ such that contraction is never applied to main
formulas of inferences in $\mathcal{R}'$.
\end{definition}

The intuitionistic system LJ is an example of a calculus in which
contraction is not admissible for all formulas. It is only complete if the main
formula of the implication left rule is contracted \cite{dyckhoff92jsl}.

The admissibility of contraction involves transformations which are similar to the
rank reduction rewriting rules of reductive cut-elimination. This is a special
case of permutation checking, since the upper inference \emph{must} be
applied to auxiliary formulas of the lower inference.

\begin{definition}[Contraction permutation]
Let $\mathbb{S}$ be a sequent calculus proof system, $c$ one of its contraction
rules and $\alpha$ a logical rule applied to a formula $F_\alpha$. We say that
$c \uparrow_c \alpha$ if a derivation composed by contraction of $F_\alpha$
followed by applications of $\alpha$ to the contracted formulas can be
transformed into a derivation where $\alpha$ is applied to $F_\alpha$ and
contraction is applied to the auxiliary formulas of $\alpha$.
\end{definition}

It is worth noting that many of the cases for contraction permutation rely on
$\alpha$ being applied to all contracted formulas in all premises where they
occur.
The proofs of such cases require a lemma stating that
$\alpha$ can be ``pushed down'' until the correct location.

If $c \uparrow_c \alpha$ for some inference $\alpha$, then it is admissible for
that inference, as it can always be replaced by contraction on its ancestors
To prove full admissibility of contraction in the
calculus, it is necessary to prove that contraction on atoms can be eliminated.
We will not address this issue in this paper, but we will analyse the behavior
of contraction among the phases in a focused proof. 

\begin{definition}[Admissibility of contraction in a phase]
Let $\mathbb{S}$ be a sequent calculus proof system and $C_1, C_2$ a focusable
permutation partition. We say that contraction is admissible in phase $i$ if
every $\mathbb{S}$ proof can be transformed into a proof where contraction is
never applied to main formulas of rules $\alpha \in C_i$.
\end{definition}

\begin{theorem}
Let $\mathbb{S}$ be a sequent calculus system,
$C$ its contraction rules,
$C_1, C_2$ a focusable permutation partition.
If for all $c \in C$ and
$\alpha \in C_i$, $c \uparrow \alpha$ and $c \uparrow_c \alpha$, then
contraction is admissible in phase $i$.
\end{theorem}




The focused proof is obtained by only contracting formulas that can be introduced
by $C_2$ rules. It is easy to extend Definition~\ref{def:focused-pf} to enforce this
as done in LJF and LKF~\cite{liang07csl}.


%


\section{Case studies}
\label{sec:cases}
Given the permutation cliques, it is up to the user to analyse them and decide
which partition to use for the focused proof system. As case studies we will show 
how the focused proof systems LKF, LJF and MALLF
can be obtained from LK, LJ and MALL respectively using the permutation cliques.

\smallskip

\textbf{MALL}
MALL stands for multiplicative additive linear logic (without exponentials) 
and its rules are depicted in Figure \ref{fig:mall}.
A focused system, MALLF, for this calculus was proposed in \cite{andreoli92jlc}.
\begin{figure}[t]
{\scriptsize
\begin{align*}
\infer[\binampersand_l]{\Gamma, A_1 \binampersand A_2 \vdash \Delta}{\Gamma, A_i \vdash \Delta} 
\;\;\;
\infer[\otimes_l]{\Gamma, A \otimes B \vdash \Delta}{\Gamma, A, B \vdash \Delta}
\;\;\;
\infer[\oplus_l]{\Gamma, A \oplus B \vdash \Delta}{\Gamma, A \vdash \Delta & \Gamma, B \vdash \Delta}
\;\;\;
\infer[\bindnasrepma_l]{\Gamma, \Gamma', A \bindnasrepma B \vdash \Delta, \Delta'}{\Gamma, A \vdash \Delta & \Gamma', B \vdash \Delta'}
\\
\infer[\binampersand_r]{\Gamma \vdash \Delta, A \binampersand B}{\Gamma \vdash \Delta, A & \Gamma \vdash \Delta, B}
\;\;\;
\infer[\otimes_r]{\Gamma, \Gamma' \vdash \Delta, \Delta', A \otimes B}{\Gamma \vdash \Delta, A & \Gamma' \vdash \Delta', B}
\;\;\;
\infer[\oplus_r]{\Gamma \vdash \Delta, A_1 \oplus A_2}{\Gamma \vdash \Delta, A_i}
\;\;\;
\infer[\bindnasrepma_r]{\Gamma \vdash \Delta, A \bindnasrepma B}{\Gamma \vdash \Delta, A, B}
\end{align*}
}
\vspace{-1cm}
\caption{MALL logical inferences}
\label{fig:mall}
\vspace{-4mm}
\end{figure}

Given the logical inferences of MALL, the permutation cliques found were the
following:
$
CL_1 = \{ \otimes_l, \oplus_l, \bindnasrepma_r, \binampersand_r, \binampersand_l, \oplus_r \}
$ and
$
CL_2 = \{ \otimes_r, \oplus_r, \bindnasrepma_l, \binampersand_l \}
$,
with the relation $CL_1 \downarrow CL_2$. The following focusable permutation partition corresponds to MALLF:
$
C_1 = \{ \otimes_l, \oplus_l, \bindnasrepma_r, \binampersand_r \}  
$ and
$
C_2 = \{ \otimes_r, \oplus_r, \bindnasrepma_l, \binampersand_l \}
$.
Notice that all invertible rules are in $C_1$, while all positive rules are in $C_2$ as expected.

\smallskip

\textbf{LK and LJ}
In order to derive the focused system LKF for classical logic from LK, all
variations of inferences must be considered. We need to take into account the
additive and multiplicative versions of each conjunction and disjunction, as
depicted in Figure \ref{fig:lk}. In principle an analysis could be made with the 
usual presentation of the LK system, but it would certainly not result in LKF. 
Asserting that we can generate a well-known focused system serves as a
validation of our method.

\begin{figure}[t]
{\scriptsize
\begin{align*}
\infer[\wedge^a_l]{\Gamma, A_1 \wedge A_2 \vdash \Delta}{\Gamma, A_i \vdash \Delta}
\;\;\;
\infer[\wedge^m_l]{\Gamma, A \wedge B \vdash \Delta}{\Gamma, A, B \vdash \Delta}
\;\;\;
\infer[\vee^a_l]{\Gamma, A \vee B \vdash \Delta}{\Gamma, A \vdash \Delta & \Gamma, B \vdash \Delta}
\;\;\;
\infer[\vee^m_l]{\Gamma, \Gamma', A \vee B \vdash \Delta, \Delta'}{\Gamma, A \vdash \Delta & \Gamma', B \vdash \Delta'}
\\
\infer[\wedge^a_r]{\Gamma \vdash \Delta, A \wedge B}{\Gamma \vdash \Delta, A & \Gamma \vdash \Delta, B}
\;\;\;
\infer[\wedge^m_r]{\Gamma, \Gamma' \vdash \Delta, \Delta', A \wedge B}{\Gamma \vdash \Delta, A & \Gamma' \vdash \Delta', B}
\;\;\;
\infer[\vee^a_r]{\Gamma \vdash \Delta, A_1 \vee A_2}{\Gamma \vdash \Delta, A_i}
\;\;\;
\infer[\vee^m_r]{\Gamma \vdash \Delta, A \vee B}{\Gamma \vdash \Delta, A, B}
\end{align*}
}
\vspace{-1cm}
\caption{Additive and multiplicative logical inferences of the LK system.}
\label{fig:lk}
\vspace{-4mm}
\end{figure}

The permutation cliques for the inferences in Figure \ref{fig:lk} are:
$
CL_1 = \{ \wedge^a_r, \wedge^m_l, \vee^m_r, \vee^a_l, \wedge^a_l, \vee^a_r \}
$ and
$
CL_2 = \{ \wedge^m_r, \wedge^a_l, \vee^a_r, \vee^m_l \}
$,
where $CL_1 \downarrow CL_2$. Analogous to MALL, we
can drop the two last rules from $CL_1$ and obtain a focusable permutation
partition which corresponds to the propositional fragment of LKF.

By analysing the permutation relation of contraction to the rules in the
partitions, we observe that it permutes up ($\uparrow$ and $\uparrow_c$) all the 
inferences in $CL_{1} \ \{\wedge^a_l, \vee^a_r \}$. Therefore, it is admissible 
in the negative phase. For the positive phase, on the other hand, contraction will
not permute up, for example, $\wedge^a_l$. We can thus conclude that such a
system must have contraction for positive formulas\footnote{This contraction is
implicit on the \emph{decide} rule and the positive rules for the usual 
presentation of LKF.}.

The case for LJ is completely analogous as that of LK when considering the
partition:
$
CL_1 = \{ \wedge^m_l, \wedge^a_r, \vee_l, \rightarrow_r, \wedge^a_l, \vee_r \} 
$ and
$
CL_2 = \{ \wedge^a_l, \wedge^m_r, \vee_r, \rightarrow_l \}
$.

\section{Conclusion}
\label{sec:conc}

This paper proposed a method for automatically devising focused proof systems
for sequent calculi. Our aim was to provide a uniform and automated way to obtain
the sound and complete systems without using an encoding in linear logic.
The main element in our solution is the permutation graph of a
sequent calculus system. By using this graph we can separate the inferences into
positives and negatives and also reason on the admissibility of contraction. The
permutation graph represents the permutation lemmas used in the proof in
\cite{miller07cslb}. We extended the method developed in \cite{miller07cslb} to
handle contraction.


For future work, we plan to apply/extend our technique to other proof systems in
order to obtain sensible focused proof systems. There are, however, some more
foundational challenges in doing so. We would need to extend the conditions used
here for determining whether a partition is focusable. For example,
non-commutative and bunched proof systems have even more complicated structural
restrictions. It is not even clear how would be the focusing discipline for
these proof systems. We expect that our existing machinery may help make some of
these decisions by investigating different partitions.


Although we can deduce in which phase the contraction of formulas is admissible,
it is still unclear if the position of this rule in the permutation graph can
indicate exactly which rules do not admit contraction. We expect to further
investigate the permutation graphs of other systems to find out if this and
other properties can be discovered.

\vspace{-0.2cm}

\bibliographystyle{eptcs}
\bibliography{leanreferences}

\begin{thebibliography}{10}
\providecommand{\bibitemdeclare}[2]{}
\providecommand{\surnamestart}{}
\providecommand{\surnameend}{}
\providecommand{\urlprefix}{Available at }
\providecommand{\url}[1]{\texttt{#1}}
\providecommand{\href}[2]{\texttt{#2}}
\providecommand{\urlalt}[2]{\href{#1}{#2}}
\providecommand{\doi}[1]{doi:\urlalt{http://dx.doi.org/#1}{#1}}
\providecommand{\bibinfo}[2]{#2}

\bibitemdeclare{article}{andreoli92jlc}
\bibitem{andreoli92jlc}
\bibinfo{author}{Jean-Marc \surnamestart Andreoli\surnameend}
  (\bibinfo{year}{1992}): \emph{\bibinfo{title}{Logic Programming with Focusing
  Proofs in Linear Logic}}.
\newblock {\sl \bibinfo{journal}{Journal of Logic and Computation}}
  \bibinfo{volume}{2}(\bibinfo{number}{3}), pp. \bibinfo{pages}{297--347},
  \doi{10.1093/logcom/2.3.297}.

\bibitemdeclare{phdthesis}{baelde08phd}
\bibitem{baelde08phd}
\bibinfo{author}{David \surnamestart Baelde\surnameend} (\bibinfo{year}{2008}):
  \emph{\bibinfo{title}{A linear approach to the proof-theory of least and
  greatest fixed points}}.
\newblock Ph.D. thesis, \bibinfo{school}{Ecole Polytechnique}.

\bibitemdeclare{incollection}{danos93wll}
\bibitem{danos93wll}
\bibinfo{author}{V.~\surnamestart Danos\surnameend}, \bibinfo{author}{J.-B.
  \surnamestart Joinet\surnameend} \& \bibinfo{author}{H.~\surnamestart
  Schellinx\surnameend} (\bibinfo{year}{1995}): \emph{\bibinfo{title}{{LKT} and
  {LKQ}: sequent calculi for second order logic based upon dual linear
  decompositions of classical implication}}.
\newblock In: {\sl \bibinfo{booktitle}{Advances in Linear Logic}}, pp.
  \bibinfo{pages}{211--224}, \doi{10.1017/CBO9780511629150.011}.

\bibitemdeclare{article}{dyckhoff92jsl}
\bibitem{dyckhoff92jsl}
\bibinfo{author}{Roy \surnamestart Dyckhoff\surnameend} (\bibinfo{year}{1992}):
  \emph{\bibinfo{title}{Contraction-free sequent calculi for intuitionistic
  logic}}.
\newblock {\sl \bibinfo{journal}{Journal of Symbolic Logic}}
  \bibinfo{volume}{57}(\bibinfo{number}{3}), pp. \bibinfo{pages}{795--807},
  \doi{10.2307/2275431}.

\bibitemdeclare{inproceedings}{liang07csl}
\bibitem{liang07csl}
\bibinfo{author}{Chuck \surnamestart Liang\surnameend} \& \bibinfo{author}{Dale
  \surnamestart Miller\surnameend} (\bibinfo{year}{2007}):
  \emph{\bibinfo{title}{Focusing and Polarization in Intuitionistic Logic}}.
\newblock In: {\sl \bibinfo{booktitle}{Computer Science Logic}}, {\sl
  \bibinfo{series}{LNCS}} \bibinfo{volume}{4646},
  \bibinfo{publisher}{Springer}, pp. \bibinfo{pages}{451--465},
  \doi{10.1007/978-3-540-74915-8\_34}.

\bibitemdeclare{article}{liang09tcs}
\bibitem{liang09tcs}
\bibinfo{author}{Chuck \surnamestart Liang\surnameend} \& \bibinfo{author}{Dale
  \surnamestart Miller\surnameend} (\bibinfo{year}{2009}):
  \emph{\bibinfo{title}{Focusing and Polarization in Linear, Intuitionistic,
  and Classical Logics}}.
\newblock {\sl \bibinfo{journal}{Theoretical Computer Science}}
  \bibinfo{volume}{410}(\bibinfo{number}{46}), pp. \bibinfo{pages}{4747--4768},
  \doi{10.1016/j.tcs.2009.07.041}.

\bibitemdeclare{inproceedings}{miller07cslb}
\bibitem{miller07cslb}
\bibinfo{author}{Dale \surnamestart Miller\surnameend} \&
  \bibinfo{author}{Alexis \surnamestart Saurin\surnameend}
  (\bibinfo{year}{2007}): \emph{\bibinfo{title}{From proofs to focused proofs:
  a modular proof of focalization in Linear Logic}}.
\newblock In: {\sl \bibinfo{booktitle}{CSL 2007}}, \bibinfo{series}{LNCS}, pp.
  \bibinfo{pages}{405--419}, \doi{10.1007/978-3-540-74915-8\_31}.

\bibitemdeclare{phdthesis}{nigam09phd}
\bibitem{nigam09phd}
\bibinfo{author}{Vivek \surnamestart Nigam\surnameend} (\bibinfo{year}{2009}):
  \emph{\bibinfo{title}{Exploiting non-canonicity in the sequent calculus}}.
\newblock Ph.D. thesis, \bibinfo{school}{Ecole Polytechnique}.

\bibitemdeclare{inproceedings}{nigam13iclp}
\bibitem{nigam13iclp}
\bibinfo{author}{Vivek \surnamestart Nigam\surnameend},
  \bibinfo{author}{Giselle \surnamestart Reis\surnameend} \&
  \bibinfo{author}{Leonardo \surnamestart Lima\surnameend}
  (\bibinfo{year}{2013}): \emph{\bibinfo{title}{Checking Proof Transformations
  with {ASP}}}.
\newblock In: {\sl \bibinfo{booktitle}{ICLP (Technical Communications)}},
  \bibinfo{volume}{13}.

\bibitemdeclare{inproceedings}{nigam14ijcar}
\bibitem{nigam14ijcar}
\bibinfo{author}{Vivek \surnamestart Nigam\surnameend},
  \bibinfo{author}{Giselle \surnamestart Reis\surnameend} \&
  \bibinfo{author}{Leonardo \surnamestart Lima\surnameend}
  (\bibinfo{year}{2014}): \emph{\bibinfo{title}{Quati: An Automated Tool for
  Proving Permutation Lemmas}}.
\newblock In: {\sl \bibinfo{booktitle}{7th {IJCAR} Proceedings}}, pp.
  \bibinfo{pages}{255--261}, \doi{10.1007/978-3-319-08587-6\_18}.

\bibitemdeclare{phdthesis}{saurin08phd}
\bibitem{saurin08phd}
\bibinfo{author}{Alexis \surnamestart Saurin\surnameend}
  (\bibinfo{year}{2008}): \emph{\bibinfo{title}{Une \'etude logique du
  contr\^ole (appliqu\'ee \`a\ la programmation fonctionnelle et logique)}}.
\newblock Ph.D. thesis, \bibinfo{school}{Ecole Polytechnique}.

\end{thebibliography}

\end{document}